\documentclass[aps,showpacs,prl,twocolumn,superscriptaddress]{revtex4-1}

\usepackage{exscale}
\usepackage{bbm}
\usepackage{graphicx}
\usepackage{amsmath}
\usepackage{latexsym}
\usepackage{amsfonts}
\usepackage{amssymb}
\usepackage{times}
\usepackage[T1]{fontenc}
\usepackage{amsthm}
\usepackage{enumerate}
\usepackage{bbold}
\usepackage{color}
\usepackage[colorlinks=true,citecolor=blue,urlcolor=blue]{hyperref}
\usepackage{mathtools}
\usepackage{changes}
\usepackage[normalem]{ulem}

\newcommand{\ba}{\begin{eqnarray}}
\newcommand{\be}{\begin{equation}}
\newcommand{\ee}{\end{equation}}
\newcommand{\beq}{\begin{equation}}
\newcommand{\eeq}{  \end{equation}}
\newcommand{\bea}{\begin{eqnarray}}
\newcommand{\eea}{  \end{eqnarray}}

\newcommand{\ea}{\end{eqnarray}}
\newcommand{\ban}{\begin{eqnarray*}}
\newcommand{\ean}{\end{eqnarray*}}
\newcommand{\Tr}{\operatorname{tr}}

\newcommand{\ket}[1]{\left|#1\right\rangle}

\newcommand{\expect}[1]{\langle#1\rangle}

\newcommand{\figref}[1]{Fig.~\ref{#1}}

\makeatletter
\newtheorem*{rep@theorem}{\rep@title}
\newcommand{\newreptheorem}[2]{%
\newenvironment{rep#1}[1]{%
 \def\rep@title{#2 \ref{##1}}%
 \begin{rep@theorem}}%
 {\end{rep@theorem}}}
\makeatother

\theoremstyle{plain}
\newreptheorem{thm}{Theorem}

\newtheorem{fakt}{Fact}

\theoremstyle{definition}

\theoremstyle{remark}

\newcommand{\fb}[1]{\textcolor{red}{[#1]}}

\begin{document}

\title{Scalable Bell inequalities for qubit graph states and robust self-testing}

\author{F. Baccari}
\affiliation{ICFO - Institut de Ciencies Fotoniques, The Barcelona Institute of Science and Technology, 08860 Castelldefels (Barcelona), Spain}

\author{R. Augusiak}
\affiliation{Center for Theoretical Physics, Polish Academy of Sciences, Aleja Lotnik\'ow 32/46, 02-668 Warsaw, Poland}

\author{I. \v{S}upi\'c}
\affiliation{ICFO - Institut de Ciencies Fotoniques, The Barcelona Institute of Science and Technology, 08860 Castelldefels (Barcelona), Spain}
\affiliation{D{\'{e}}partement de Physique Appliqu\'{e}e, Universit\'{e} de Gen\`{e}ve, 1211 Gen\`{e}ve, Switzerland}

\author{J. Tura}  
\affiliation{Max-Planck-Institut f\"ur Quantenoptik, Hans-Kopfermann-Stra{\ss}e 1, 85748 Garching, Germany}

\author{A. Ac\' in}
\affiliation{ICFO - Institut de Ciencies Fotoniques, The Barcelona Institute of Science and Technology, 08860 Castelldefels (Barcelona), Spain}
\affiliation{ICREA - Institucio Catalana de Recerca i Estudis Avan\c cats, Pg. Lluis Companys 23, 08010 Barcelona, Spain}

\begin{abstract}
Bell inequalities constitute a key tool in quantum information theory: they not only allow one to reveal nonlocality in composite quantum systems, but, more importantly, they can be used to certify relevant properties thereof. We provide a very simple and intuitive construction of Bell inequalities that are maximally violated by the multiqubit graph states and can be used for their robust self-testing. The main advantage of our inequalities over previous constructions for these states lies in the fact that the number of correlations they contain scales only linearly with the number of observers, which presents a significant reduction of the experimental effort needed to violate them. We also discuss possible generalizations of our approach by showing that it is applicable to entangled states whose stabilizers are not simply tensor products of Pauli matrices. 
\end{abstract}

\date{\today}

\maketitle

\textit{Introduction.} Since the first discovery by Bell more than 50 years ago, it has been known that quantum particles are capable of sharing correlations that cannot be reproduced by any classical means \cite{bell1964on,brunner2014bell}. Such property, usually referred to as nonlocality, has attracted a lot of attention, not only because of its fundamental interest, but also because of its applications in quantum technologies.
In particular, nonlocality is a key ingredient in the framework of device-independent protocols, among which the best known applications are in the field of cryptography and randomness certification 	\cite{mayers1998quantum,acin2007device,colbeck2006quantum,pironio2010random,acinmasanes2016}.

Another relevant application of nonlocality is self-testing, which can be seen as a way of certifying both the state produced and the local measurements performed by some given quantum devices, by simply looking at the resulting correlations. Such a tool is particularly interesting because it offers a way to guarantee that the devices are working properly without the need of knowing their internal functioning. It thus consitutes a form of device-independent certification that can be useful for various quantum information protocols. In fact, since its introduction in \cite{mayers2004self-testing}, self-testing has been studied in many contexts, showing to be applicable to multipartite states \cite{mckague2011self,supic2017simple} and any number of measurements as well \cite{supic2016self}. Moreover, extensions to several different scenarios have also been considered, such as steering \cite{supic2016steering}, prepare-and measure framework \cite{tavakoli2018self}, networks \cite{bowles2018device,bowles2018self} and the certification of quantum channels \cite{sekatski2018certifying}.

From an implementation perspective, a relevant challenge is to design self-testing strategies that can be applied to realistic situations. Since recent experiments are capable of addressing already tens of particles \cite{wang2018qubit,friis2018observation} a crucial ingredient for a certifying strategy is to present an efficient scaling in terms of the required resources. Indeed, any method that is based on the full information about either the state or the observed correlations is bound to become intractable already for medium-large systems, since such information scales exponentially with the number of particles involved.
Interestingly, it has already been shown that nonlocality can be assessed with the knowledge of few-body correlations only \cite{tura2014detecting,tura2015nonlocality,tura2014translationally,tura2017energy,baccari2017efficient}, which require generally a polynomial scaling number of measurements to be estimated. Moreover, Bell inequalities that consist of a constant amount of terms have also been introduced \cite{tura2014detecting,tura2015nonlocality}, opening the way to the first experimental detections of Bell correlations in many-body systems of hundreds \cite{schmied2016bell} and hundreds of thousands of atoms \cite{engelsen2017bell}.

Since nonlocality is a necessary ingredient for self-testing, a relevant question to ask is which family of multipartite states can be self-tested using a polynomial amount of information about the observed correlations. Here, we address this question by focusing on graph states, one of the most representative subsets of multipartite entangled states that include, for instance, Greenberger-Horne-Zeilinger (GHZ) and cluster states. In particular, we introduce the first scalable self-testing method for graph states based on Bell inequalities. This implies introducing a new family of Bell inequalities maximally violated by graph states whose violation, contrary to previous constructions \cite{guhne2005bell,toth2006twosetting,guhne2008generalized}, can be estimated by measuring a number of correlations that scales linearly with the number of particles. While other works have already proven self-testing for these states with a similar amount of information \cite{mckague2011self,supic2017simple} the main novelty of our result is its connection to the violation of a Bell inequality, which is an advantage that can be exploited for further applications. Here we focus on two: first, we apply the techniques from \cite{kaniewski2016analytic} to show that our inequalities can self-test the graph states in a robust way, and second, we analyse generalisations of our method that allow us to derive Bell inequalities useful for self-testing  multiqubit states that are not stabilizer states.



\textit{Preliminaries.} Before presenting our results, we first set up the scenario and introduce the relevant notation and terminology. We consider the simplest $N$-partite Bell scenario, referred to as $(N,2,2)$ scenario, in which $N$ distant observers share some $N$-partite quantum state $\ket{\psi}$. On their share of the state, observer $i$ measures one of two dichotomic observables, denoted $A_{x_i}^{(i)}$ with $x_i=0,1$ and $i=1,\ldots,N$, whose outcomes are labelled $\pm 1$ (in the few-party case we will also denote the observables by $A_i$, $B_i$ etc.). 
The correlations obtained in this experiment are described by a collection of
expectation values
\begin{equation}\label{eq:expvalues}
\langle A_{x_{i_1}}^{(i_1)}\ldots A_{x_{i_k}}^{(i_k)}\rangle=\langle \psi|A_{x_{i_1}}^{(i_1)}\otimes \ldots\otimes A_{x_{i_k}}^{(i_k)}|\psi\rangle
\end{equation}
which are usually referred to as correlators and we arrange them for our convenience in a vector $\vec{c}$. 
Such quantum correlations form a convex set, denoted $Q_N$.
%
Noticeably, it contains correlations that, even if obtained from a quantum state and quantum measurements, can be simulated in a purely classical way. Such correlations are said to admit a local hidden variable (LHV) model and are shortly called \textit{local} or \textit{classical}. They form a convex polytope, denoted $P_N$. Yet, quantum theory offers also correlations that escape the description in terms of LHV models and Bell was the first to reveal that \cite{bell1964on}. To this end, he used certain inequalities---so-called Bell inequalities---whose general form in the $(N,2,2)$ scenario is 
\begin{equation}\label{eq:Bellineq}
I:=\sum_{k=1}^N\sum_{\substack{1 \leq i_1<i_2<\ldots<i_k \leq N  \\  x_{i_1},\ldots,x_{i_N}=0,1 }} \alpha_{x_1, \ldots x_N}^{i_1, \ldots i_k} \langle A_{x_{i_1}}^{(i_1)}\ldots A_{x_{i_k}}^{(i_k)}\rangle\leq \beta_C
\end{equation}
with $\beta_C$ being the classical bound defined as $\beta_C=\max_{P_N}I$. 
Correlations $\vec{c}$ that violate a Bell inequality cannot be reproduced by any LHV model and are therefore termed nonlocal. The best known example of a Bell inequality, defined in the $(2,2,2)$ scenario, is known as Clauser-Horne-Shimony-Holt (CHSH)~\cite{clauser1969proposed} and reads
\begin{equation}\label{eq:CHSH}
I_{\mathrm{CHSH}}:= \expect{(A_0+A_1) B_0} + \expect{(A_0-A_1) B_1} 
\leq 2,
%
\end{equation}
where $A_x$ and $B_y$ $(x,y=0,1)$ are dichotomic observables measured by the respective observers. Its maximal quantum value is $2 \sqrt{2}$ and is achieved by the maximally 
entangled state of two qubits $\ket{\phi_+}=(\ket{00}+\ket{11})/\sqrt{2}$ and the 
observables $A_i=[\sigma_X+(-1)^i\sigma_Z]/\sqrt{2}$, and $B_0=\sigma_X$ and $B_1=\sigma_Z$. Here, $\sigma_X$ and $\sigma_Z$ are the Pauli operators. 


Let us finally recall the definition of the multi-qubit graph states. Consider a graph $G=(V,E)$, where $V$ is the set of vertices of size $|V|=N$, and $E$ is the set of edges connecting the vertices. Let then $n(i)$ be the neighbourhood of 
the vertex $i$, i.e., all vertices from $V$ that are connected with $i$ by an edge. Now, to every vertex $i$ we associate an operator
\begin{equation}
G_i={\sigma_X}_i\otimes \bigotimes_{j\in n(i)} {\sigma_Z}_j,
\end{equation}
in which the $\sigma_X$ operator acts on site $i$, while the $\sigma_Z$ operators 
act on all sites that belong to $n(i)$.
Then, the graph state $\ket{\psi_G} $ associated to $G$ is defined as the unique eigenstate of all these operators $G_i$ $(i=1,\ldots,N)$ with eigenvalue one. The $G_i$'s are called stabilizing operators of $\ket{\psi_G}$ and they generate $2^N$-element commutative group of operators stabilizing $\ket{\psi_{G}}$, called \textit{stabilizer group}. 

The simplest example of a graph state, corresponding to the two-vertex complete graph up to local unitary equivalence, is precisely the maximally entangled state of two qubits $\ket{\phi_+}$. Let us stress here that there is a direct relation between the stabilizing operators of $\ket{\phi_+}$ and the maximal quantum violation of the CHSH Bell inequality. Precisely, the observables realising the maximal quantum violation of the CHSH Bell inequality can be combined to obtain $(A_0+A_1)\otimes B_0=\sqrt{2}\sigma_X\otimes \sigma_X$ and $(A_0-A_1)\otimes B_1=\sqrt{2}\sigma_Z\otimes \sigma_Z$, 
which constitute, up to a constant factor, the stabilizing operators of $\ket{\phi_+}$. This is exactly the relation that we exploit below to construct Bell inequalities for graph states. In fact, our inequalities can be seen as a generalisation of the CHSH Bell inequality to the multipartite case.


\textit{CHSH-like Bell inequalities for graph states.}
%
We are now ready to present our main results. We begin with a new family of Bell inequalities maximally violated by the graph states. Let us consider a graph $G$ and, for our convenience, enumerate its vertices so that the first one has the largest neighbourhood, that is, $| n(1) | =\max_{i} | n(i) | \equiv n_{\max}$. If there are more vertices with maximal neighbourhood in $G$, we choose any of them as the first vertex. 

Then, to every stabilizing operator $G_i$ corresponding to the graph $G$ we associate an expectation value in which the respective operators are replaced by quantum dichotomic observables or combinations thereof. More precisely, at the first site ${\sigma_X}_1$ and ${\sigma_Z}_1$ are replaced by, respectively, $A_0^{(1)}+A_1^{(1)}$ and $A_0^{(1)} - A_1^{(1)}$, whereas for the remaining observers, ${\sigma_X}_j$ and ${\sigma_Z}_j$ are replaced simply by $A_0^{(j)}$ and $A_1^{(j)}$.
Finally, if there is an identity at any position in $G_i$ we leave it as it is. 

We then add the obtained correlators, multiplying the first one by $n_{\max}$, and obtain the following Bell inequality
\begin{eqnarray}\label{eq:inequalities}
I_G:&=&n_{\max}\left\langle (A_0^{(1)}+A_1^{(1)})\prod_{i\in n(1)}A_{1}^{(i)}\right\rangle\nonumber\\
&&+\sum_{i\in n(1)}
\left\langle (A_0^{(1)}-A_1^{(1)})A_0^{(i)}\prod_{j\in n(i)\setminus\{1\}}A_1^{(j)}\right\rangle\nonumber\\
&&+\sum_{i\notin n(1)\cup\{1\}}
\left\langle A^{(i)}_0 \prod_{j\in n(i)} A^{(j)}_{1}\right\rangle\leq \beta_G^C .
\end{eqnarray}
%
%
%
Notice that $I_G$ coincides with the CHSH Bell expression for $N = 2$. Similarly, for higher $N$ it can be seen as a sum of $n_{\max}$ CHSH Bell expressions between the first party and some joint measurements on the parties corresponding to the neighbouring vertices, plus some number of correlators in which the first observer does not appear. This simple structure makes our Bell inequalities 
extremely easy to characterize. In fact, as shown below, their maximal classical and quantum values can even be computed by hand. 
\begin{fakt}\label{Obs1}For a given graph $G$,
the classical bound of the corresponding Bell inequality (\ref{eq:inequalities})
is $\beta_G^C=N+n_{\max}-1$. 
\end{fakt}
\begin{proof}We start by noting that (\ref{eq:inequalities}) consists of a 
single term containing $A_0^{(1)}+A_1^{(1)}$ appearing with weight $n_{\max}$, 
and $n_{\max}$ different terms containing  $A_0^{(1)}-A_1^{(1)}$. Now, for any 
local deterministic correlations that assign $\pm 1$ to all observables $A_{x_j}^{(j)}$, these two combinations are either zero or two and if one equals two, the other vanishes. Thus, the contribution from these terms to the classical bound is exactly $2n_{\max}$. Then, the maximal value of the remaining correlators in (\ref{eq:inequalities}) over local deterministic strategies is $N-n_{\max}-1$, which together with the first
contribution results in $\beta_G^C=N+n_{\max}-1$. 
\end{proof}

\begin{fakt}\label{Obs2}
For a given graph $G$, the maximal quantum violation of (\ref{eq:inequalities})
is $\beta_G^Q=(2\sqrt{2}-1)n_{\max}+N-1$. 
\end{fakt}

\begin{proof} We first demonstrate that $\beta_{G}^Q$ upper bounds the maximal quantum value of (\ref{eq:inequalities}) and then provide an explicit quantum strategy achieving this bound.

Let us consider a Bell operator $\mathcal{B}_G$ obtained from the Bell expression $I_G$. 
It can be shown that the shifted Bell operator $\beta_G^Q\mathbbm{1}-\mathcal{B}_G$ is positive semidefinite for any choice of local observables $A^{(i)}_{x_i}$, 
meaning that $\beta_{G}^Q$ upper bounds the maximal quantum 
violation of inequality (\ref{eq:inequalities}). This is achieved by writing this operator in the form of a sum of squares decomposition (see Appendix~A).
To complete the proof let us now provide an explicit quantum strategy 
for which the value of $I_G$ equals exactly $\beta_G^Q$. To this end, 
we choose the following observables
%
%
$A_0^{(1)}=(\sigma_X+\sigma_Z)/\sqrt{2}$ and $A_1^{(1)}=(\sigma_X-\sigma_Z)/\sqrt{2}$ for the first observer, and
$A_0^{(i)}=\sigma_X$ and $A_1^{(i)}=\sigma_Z$ for the remaining ones.
%
By the very definition of the state $|\psi_G\rangle$ corresponding to the graph $G$, 
it is not difficult to see that for these observables and $|\psi_{G}\rangle$, the value of every correlator in (\ref{eq:inequalities}) containing combinations of the observables $A_{x_1}^{(1)}$ is $\sqrt{2}$, whereas the value of each of the remaining correlators is one. Consequently, $I_G$ for this realisation amounts to $2\sqrt{2}\,n_{\max}+N-n_{\max}-1=(2\sqrt{2}-1)n_{\max}+N-1$, which is exactly $\beta_G^Q$.
\end{proof}

A few comments are in order. First, it follows that for any graph $G$, our Bell inequalities are nontrivial, i.e., $\beta_G^Q>\beta_G^C$. On the other hand, the ratio $\beta_G^Q/\beta_G^C$ tends to a constant value (also when $n_{\max}$ depends on $N$).

\begin{figure}[t!!!]
\centering
\includegraphics[width=\columnwidth]{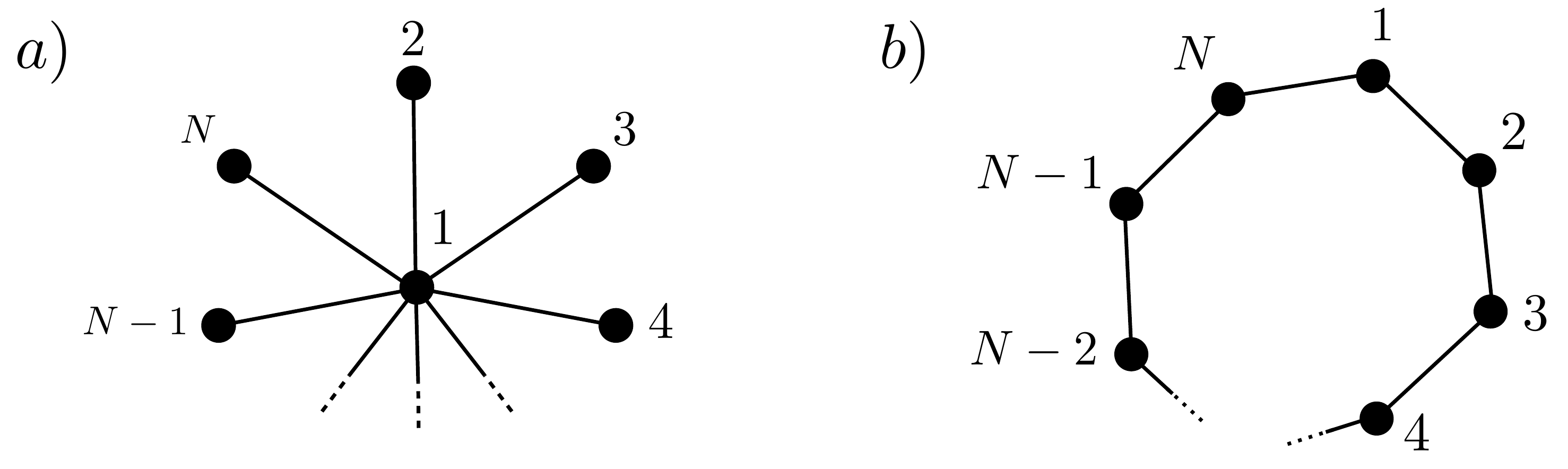}
\caption{Two examples of graphs: (a) the star graph and (b) the ring graph.} 
\label{fig:examples}
\end{figure}

\textit{Examples.}
Let us now illustrate our construction with two examples.
The first one concerns the star graph presented in Fig. \ref{fig:examples}. For this graph $| n(1) | =n_{\mathrm{max}}=N-1$ and the stabilizing operators are of the form: $G_1={\sigma_X}_1 {\sigma_Z}_2\ldots {\sigma_Z}_N$ for the first vertex and $G_i={\sigma_X}_i {\sigma_Z}_1$ with $i=2,\ldots,N$ for the remaining ones. For our convenience, we apply the Hadamard gate to all the vertices but the first one, which gives us an equivalent set of operators: $G_1'={\sigma_X}_1\ldots {\sigma_X}_N$ and $G_i'={\sigma_Z}_1 {\sigma_Z}_i$ with $i=2,\ldots,N$. It follows that they stabilise the $N$-qubit GHZ state
\begin{equation}\label{GHZstate}
|\mathrm{GHZ}_N\rangle=\frac{1}{\sqrt{2}}(\ket{0}^{\otimes N}+\ket{1}^{\otimes N}).
\end{equation}

Let us then associate expectation values to each $G_i'$. As the first vertex is the one with the largest neighbourhood, we make the assignments 
\begin{eqnarray}
G_1'&\to& \langle(A_0^{(1)}+A_1^{(1)})A_0^{(2)}\ldots A_0^{(N)}\rangle \nonumber\\
G_i'&\to& \langle (A_0^{(1)}-A_1^{(1)})A_1^{(i)}\,\rangle \qquad (i=2,\ldots,N)
\end{eqnarray}
%
%
which leads us to the following Bell inequality
\begin{eqnarray}\label{eq:GHZineq}
\mathcal{I}_{GHZ}^{N} &= & (N-1)\left[  \expect{A_{0}^{(1)}  A_{0}^{(2)} {\ldots}  A_{0}^{(N)}} + 
\expect{A_{1}^{(1)}  A_{0}^{(2)} {\ldots}  A_{0}^{(N)}} )\right] \nonumber\\
&&+\sum_{i = 2}^{N} ( \expect{A_{0}^{(1)}  A_{1}^{(i)}} - \expect{A_{1}^{(1)}  A_{1}^{(i)}} )   \leq 2(N-1).
\end{eqnarray}
This inequality was also found in Ref. \cite{baccari2017efficient} using a different approach
and it can be seen as a sum of $N-1$ CHSH Bell inequalities between the first observer
and the remaining ones; for $N=2$ it reproduces the CHSH inequality. It follows from Fact \ref{Obs2} that $\beta_{\mathrm{GHZ}}^Q=2\sqrt{2}(N-1)$ and it is achieved by the GHZ state (\ref{GHZstate}). It should be noticed that contrary to the well-known Mermin Bell inequality \cite{mermin1990extreme} which is also maximally violated by this state, our inequality contains a number of correlators that scales linearly with $N$. Moreover, only two of them are $N$-body and they involve two different measurements choices only for the first party. All this makes our inequality for the GHZ state more advantageous from the experimental point of view.

As a second example we consider the ring graph presented in Fig. \ref{fig:examples}, 
for which the stabilizing operators are $G_i={\sigma_Z}_{i-1}{\sigma_X}_i{\sigma_Z}_{i+1}$ with $i=1,\ldots,N$, where
we use the convention that ${\sigma_Z}_0\equiv {\sigma_Z}_N$ and ${\sigma_Z}_{N+1}\equiv {\sigma_Z}_1$.

As every vertex in this graph has neighborhood of the same size, i.e., 
$n(i)=n_{\max}=2$ $(i=1,\ldots,N)$, we choose the first vertex to be the one at which 
we introduce combinations of observables. Thus, following our recipe, 
\begin{eqnarray}
G_N&\to& \langle A_1^{(N-1)}A_0^{(N)}(A_0^{(1)}-A_1^{(1)})\rangle\nonumber\\
G_1&\to& \langle A_1^{(N)}(A_0^{(1)}+A_1^{(1)})A_1^{(2)}\rangle\nonumber\\
G_2&\to& \langle (A_0^{(1)}-A_1^{(1)})A_0^{(2)}A_1^{(3)}\rangle
\end{eqnarray}
and 
%
$G_i\to \langle A_{1}^{(i-1)}A_{0}^{(i)}A_1^{(i+1)}\rangle$
%
for $i=3,\ldots,N-1$. These expectation values give rise to the following Bell inequality
\begin{eqnarray}\label{eq:ring}
I_{\mathrm{ring}}&:=&2\langle A_1^{(N)}(A_0^{(1)}+A_1^{(1)})A_1^{(2)}\rangle+
\langle (A_0^{(1)}-A_1^{(1)})A_0^{(2)}A_1^{(3)}\rangle\nonumber\\
&&+\langle A_1^{(N-1)}A_0^{(N)}(A_0^{(1)}-A_1^{(1)})\rangle\nonumber\\
&&+\sum_{i=3}^{N-1}\langle A_1^{(i-1)}A_0^{(i)}A_1^{(i+1)} \rangle\leq N+1,
\end{eqnarray}
whose classical bound stems directly from Fact \ref{Obs1}, while, according to Fact \ref{Obs2}, its maximal quantum violation is $N+4\sqrt{2}-3$ and is achieved by the so-called $N$-qubit ring cluster state stabilized by the above $G_i$. Remarkably, this Bell inequality contains
only three-body nearest-neighbour correlators, i.e., correlators of minimal length able to 
detect nonlocality of the ring state \footnote{Recall that as proven in Ref. \cite{gittsovich2010multiparticle} one cannot detect entanglement of graph states only from their two-body marginals as they are compatible with separable states.}.
Lastly, notice that in this second example the ratio $\beta_G^Q/\beta_G^C$ tends to $1$ in the limit of large $N$, making the violation very sensitive to experimental errors for big systems. This issue can be fixed by properly modifying the inequality with the addition of substitutions ${\sigma_X}_j,{\sigma_Z}_j \rightarrow  (A_0^{(j)} \pm A^{(j)}_1)$ on other vertices $j$ whose neighbourhood doesn't overlap with $n(1)$ (see Appendix D for a detailed explanation of the generalised method). 
\begin{figure*}[ht!!!]
\centering
\includegraphics[width=0.4\textwidth]{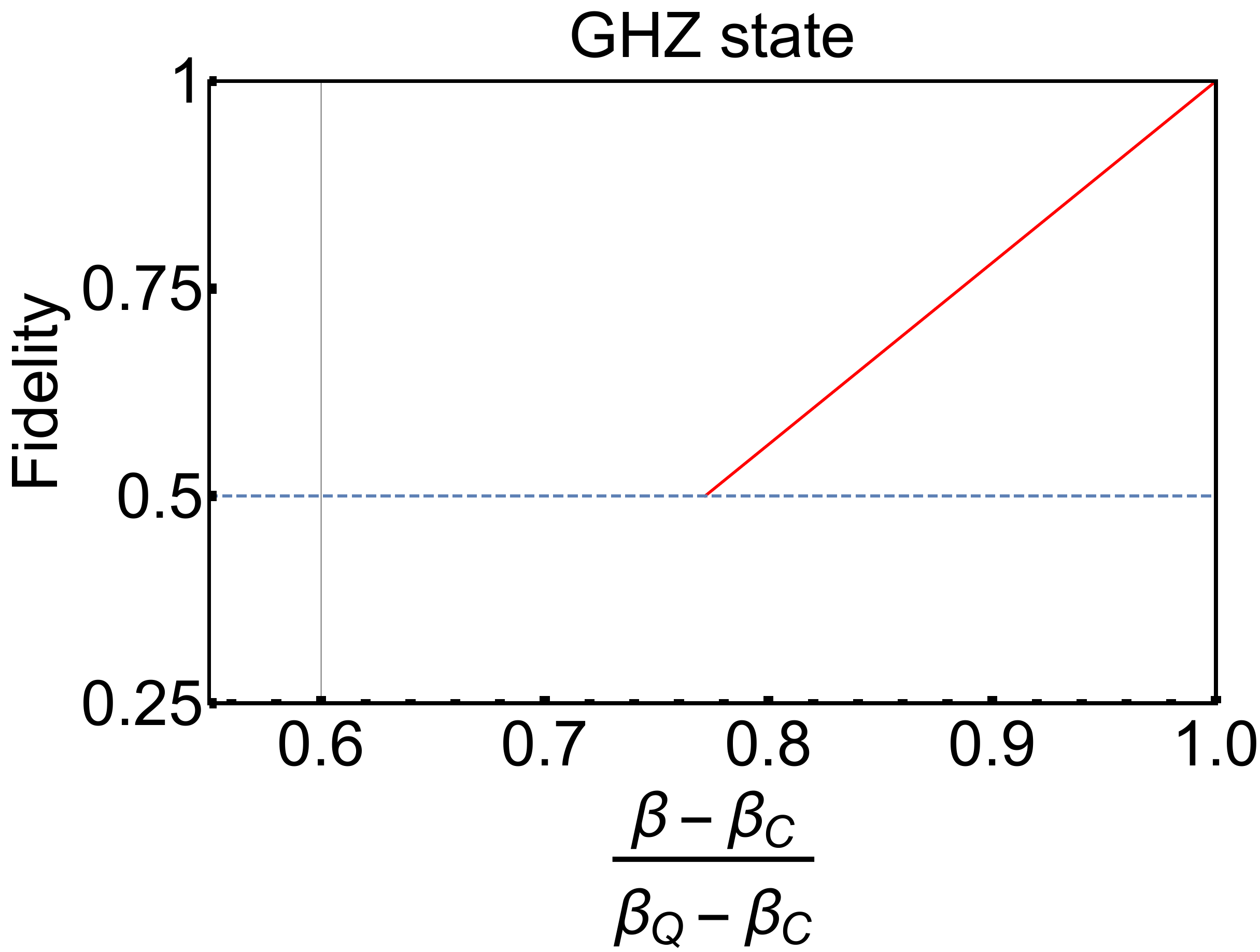}
\includegraphics[width=0.4\textwidth]{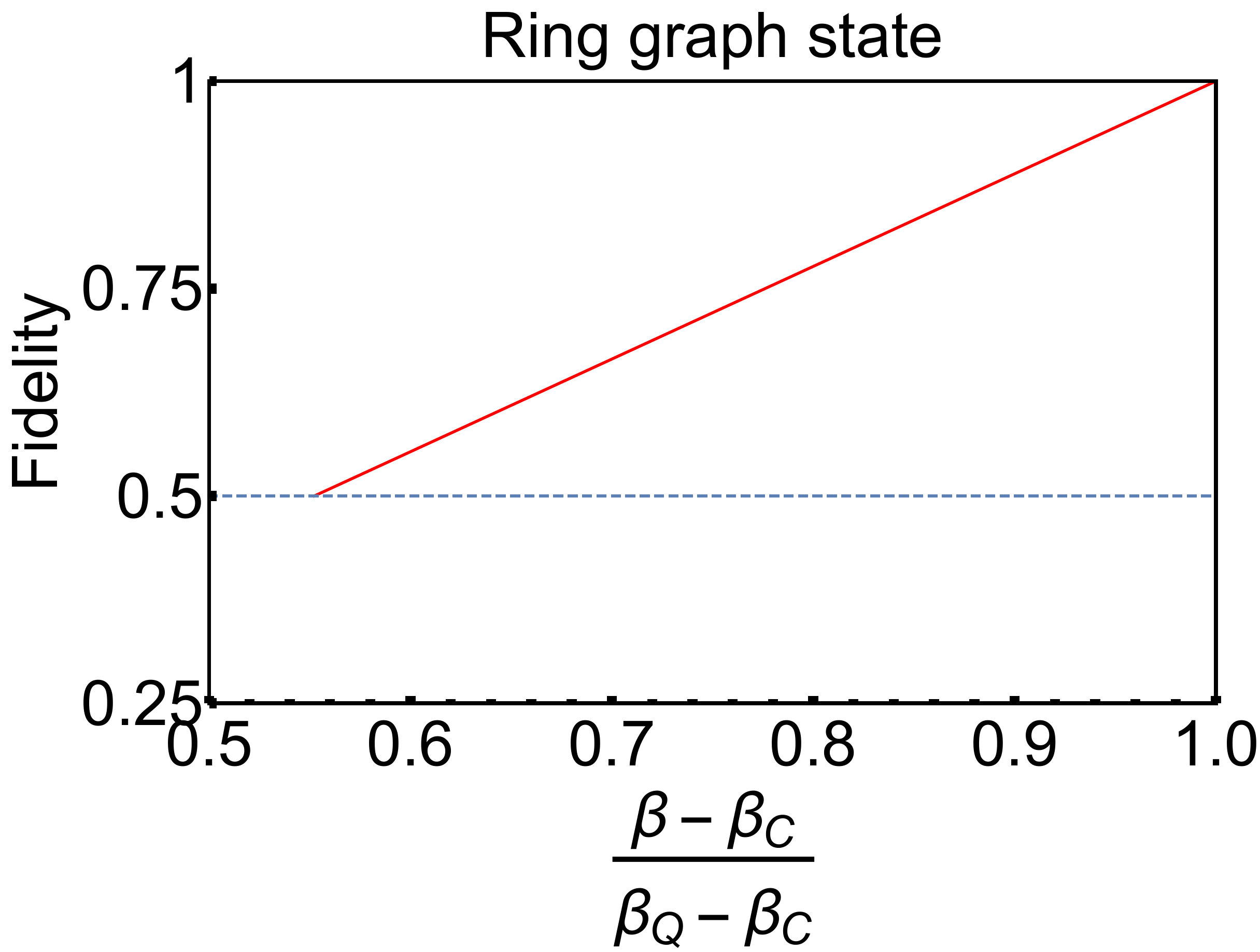}
\caption{Fidelity with the target graph state, numerically estimated as a function of the relative observed violation $(\beta -  \beta_C)/(\beta_Q - \beta_C)$ of the corresponding Bell inequality constructed with our method. The plots show the case of a GHZ state (left) and ring graph state (right) of $N = 7$ particles.} 
\label{fig:plots}
\end{figure*}

\textit{Comparison to other constructions.} Let us compare our inequalities to previous constructions of Bell inequalities for graph states. The most general one was introduced in \cite{guhne2005bell} and then modified in \cite{toth2006twosetting} to allow for two measurements at all sites. 
%
%
One of the key properties of the inequalities of Refs. \cite{guhne2005bell,toth2006twosetting} is that the ratio between their maximal quantum and classical values is exponential in $N$, making them robust against experimental imperfections. 
However, this last feature is only possible due to the fact that the amount of expectation values they contain grows exponentially with $N$, which certainly makes them highly impractical for experiments involving large number of parties. In contrast, our Bell inequalities have a much simpler structure. In particular, they require measuring only $N-n_{\max}-1$ expectation values, which results in an exponential reduction in the experimental effort needed to violate them. The price to pay is, however, that the ratio $\beta_G^Q / \beta_G^C$ tends to a constant for large $N$.

\textit{DI applications. Self-testing.} Apart from being convenient from the experimental point of view, our Bell inequalities also find applications in self-testing. To recall the task of self-testing, imagine a quantum device that performs a Bell test with some quantum state $|\bar\psi\rangle$ and quantum observables $\bar{A}_{x_i}^{(i)}$, producing correlations $\vec{c}$. The aim of self-testing is to reveal the structure of the system $\{|\bar\psi\rangle,\bar{A}_{x_i}^{(i)}\}$ from the violation of the Bell inequality by the observed correlations $\vec{c}$ 

Indeed, we can prove the following fact:

\begin{fakt}\label{Fact3}
Given a graph $G$, if the corresponding Bell inequality (\ref{eq:inequalities})
is violated maximally by a state $|\psi\rangle$ and observables $\bar{A}_j^{(i)}$, then the following holds true:
$|\psi\rangle$ is equivalent, up to local isometries, to the state $\ket{\psi_G}$ associated to the graph $G$ and, similarly, the observables are equivalent to $A_j^{(1)}=[\sigma_X+(-1)^{j}\sigma_Z]/\sqrt{2}$, and $A_0^{(i)}=\sigma_X$ and $A_1^{(i)}=\sigma_Z$ for $i=2,\ldots,N$.
\end{fakt}

%
%
%
%

\begin{proof}
The proof is in Appendix C.
\end{proof}

It should be noticed here that, compared to other self-testing methods for graph states, we present the first method that exploits the maximal violation of a multipartite Bell inequality. 
Moreover, with the aid of the approach developed in Ref. \cite{kaniewski2016analytic}
our Bell inequalities allow one to make robust self-testing statements. In fact, 
the numerical results in Fig. \ref{fig:plots} show that the fidelity between 
the state $|\psi\rangle$ violating our inequalities for two exemplary graphs
and the corresponding graph state $\ket{\psi_G}$ is a linear function of the value $I_G$.


\textit{Generalizations.} Interestingly, our construction can be generalized so to work in cases where the stabilizer operators are not products of Pauli operators. 
To give an example, let us consider the partially entangled GHZ state
\begin{equation}\label{tiltedGHZ}
|\mathrm{GHZ}_N(\theta)\rangle=\cos\theta\ket{0}^{\otimes N}+\sin\theta\ket{1}^{\otimes N} \, ,
\end{equation}
with $\theta\in(0,\pi/4]$, which is stabilized by $S_1 = \sin{2\theta} {\sigma_X}_1 {\sigma_X}_2 \ldots {\sigma_X}_N +  \cos{2 \theta} {\sigma_Z}_1$ and $S_i = {\sigma_Z}_1 {\sigma_Z}_i$ $(i = 2,\ldots, N)$. To construct a Bell
inequality maximally violated by (\ref{tiltedGHZ}), we associate, as before, an 
expectation value to each stabilizing operator $S_i$: at the first site we substitute
\begin{equation}
{\sigma_X}_1 \to \frac{A_0^{(1)} + A_1^{(1)}}{2\sin{\mu}},\qquad
{\sigma_Z}_1 \to \frac{A_0^{(1)} - A_1^{(1)}}{2\cos{\mu}},
\end{equation}
whereas at the remaining sites we traditionally set ${\sigma_X}_i\to A_0^{(i)}$ and ${\sigma_Z}_i\to A_1^{(i)}$. For $2\sin^2\mu=\sin^22\theta$, we obtain the following Bell inequality
\begin{eqnarray}
\hspace{-1cm}\mathcal{I}_{
\theta}&:=&(N-1)\langle(A_0^{(1)}+A_1^{(1)})A_0^{(2)}\ldots A_0^{(N)}\rangle\nonumber\\
&&+(N-1)\frac{\cos{2\theta}}{\sqrt{1 + \cos^2{2\theta}}}(\langle A_0^{(1)}\rangle-\langle A_1^{(1)}\rangle)\nonumber\\
&&+\frac{1}{\sqrt{1 + \cos^2{2\theta}}}\sum_{i=2}^N \langle (A_0^{(1)}- A_1^{(1)})A_1^{(i)}\rangle\leq \beta_C.
\end{eqnarray}
In Appendix E we prove that this inequality is maximally violated by the state (\ref{tiltedGHZ}) and that it can be used to self-test this state for any $\theta \in (0,\pi/4]$. Noticeably, the case $\theta = \pi/4$ recovers the inequality \eqref{eq:GHZineq} for the GHZ state. Let us also notice that for $N=2$ we obtain a Bell inequality maximally violated by any pure entangled state, which is inequivalent to the well-known tilted CHSH Bell inequality \added{\cite{acinmassarpironio,bamps2015sum}}. 

\textit{Conclusion.} We have introduced a family of Bell inequalities that are maximally violated by the graph states and are scalable from an experimental point of view. That is, contrary to the previous constructions of Bell inequalities for graph states, the number of expectation values they contain grows only linearly with the number of parties. Furthermore, the extremely simple structure of our Bell inequalities makes them easily applicable to robust self-testing.

It is worth pointing out that in constructing our inequalities we follow an approach that, similarly to those developed in Refs. 
\cite{salavrakos2017bell,kaniewski2018maximal}, exploits the quantum properties of the states and
measurements to be self-tested, rather than 
the standard approach based on the geometry of the set of 
local correlations. 

Our considerations provoke
%
%
further questions. First, it would be interesting to see whether
the approach presented here can be generally applied to entangled states stabilized by 
operators that are not just products of Pauli matrices. Here we showed that 
such generalization is possible for partially entangled GHZ states. Second, it 
would be of interest to investigate whether this approach can be exploited 
for multipartite states of higher local dimensions; in particular, the multiqudit graph states, for which no Bell inequalities are known.
Let us also mention that another method to derive Bell inequalities from stabilizing formalism was presented in \cite{sekatski2018certifying}. It would be of interest to explore possible connections between both approaches.

\textit{Acknowledgments.} 
R.~A.~acknowledges the support from the Foundation for Polish Science through the First Team project (First TEAM/2017-4/31) co-financed by the European Union under the European Regional Development Fund. I.~\v{S}., F. B. and A. A. acknowledge the support from Spanish MINECO (QIBEQI FIS2016-80773-P, Severo Ochoa SEV-2015-0522 and a Severo Ochoa PhD fellowship), Fundacio Cellex, Generalitat de Catalunya (SGR1381 and CERCA Program), ERC CoG QITBOX and AXA Chair in Quantum Information Science. I. \v{S}. also acknowledges the support from SNF (Starting grant DIAQ) and COST project CA16218, NANOCOHYBRI. This project has received funding from the European Union's Horizon 2020 research and innovation programme under the Marie-Sk\l{}odowska-Curie grant agreement No 748549.


%

\appendix

\section{Appendix A: Sum of squares decomposition}
\label{AppA}

Here we provide a more detailed proof of Fact \ref{Obs2}. 
For completeness let us also state the fact here.

\setcounter{fakt}{1}

\begin{fakt}\label{app:Obs2}
For a given graph $G$, the maximal quantum violation of (\ref{eq:inequalities})
is $\beta_G^Q=(2\sqrt{2}-1)n_{\max}+N-1$. 
\end{fakt}

\begin{proof}Let us consider dichotomic observables 
$A_{x_i}^{(i)}$ for each observer and construct from them the 
Bell operator corresponding to the Bell expression $I_G$, 
%
%
%
%
\begin{eqnarray}
\mathcal{B}_G:&=&n_{\max}(A_0^{(1)}+A_1^{(1)})\otimes \bigotimes_{i\in N(1)}A^{(i)}_{1}\nonumber\\
&&+\sum_{i\in N(1)}
(A_0^{(1)}-A_1^{(1)})\otimes A_0^{(i)}\otimes\bigotimes_{j\in N(i)\setminus\{1\}} A^{(j)}_1\nonumber\\
&&+\sum_{i\notin N(1)\cup\{1\}}
 A^{(i)}_0\otimes\bigotimes_{j\in N(i)} A^{(j)}_{1},
\end{eqnarray}
By direct checking it is not difficult to see
that the shifted Bell operator
$\beta_G^Q\mathbbm{1}-\mathcal{B}_G$ 
can be decomposed into the following sum of squares 
\begin{eqnarray}\label{app:SOS}
\beta_G^Q\mathbbm{1}-\mathcal{B}_G
&=&\frac{n_{\max}}{\sqrt{2}}
\left(\mathbbm{1}-P_1\right)^2
+\frac{1}{\sqrt{2}}\sum_{i\in N(1)}\left(\mathbbm{1}-P_i\right)^2\nonumber\\
&&+\frac{1}{2}\sum_{i\notin N(1)\cup\{1\}}
\left(\mathbbm{1}-P_{i}\right)^2,
\end{eqnarray}
where $P_i$ are operators defined as
\begin{equation}
P_1=\frac{A_0^{(1)}+A_1^{(1)}}{\sqrt{2}}\otimes \bigotimes_{i\in N(1)} A_{1}^{(i)},
%
\end{equation}
\begin{equation}
P_i=\frac{A_0^{(1)}-A_1^{(1)}}{\sqrt{2}}\otimes
A_0^{(i)}\otimes\bigotimes_{j\in N(i)\setminus\{1\}} A^{(j)}_1
\end{equation}
for $i\in N(1)$, and, finally,
\begin{equation}
P_i =A^{(i)}_0\otimes\bigotimes_{j\in N(i)} A^{(j)}_{1}
\end{equation}
for $i\notin N(1)\cup\{1\}$. This immediately implies that $\beta_G^Q\mathbbm{1}-\mathcal{B}_G \succeq 0$, and since the decomposition (\ref{app:SOS}) holds true for any 
choice of local observables $A_{x_i}^{(i)}$, we have that 
$\beta_G^Q$ upper bounds the maximal quantum value of $I_G$, that is,
\begin{equation}\label{app:upper}
\max_{Q_N}I_G\leq \beta_G^{Q}.
\end{equation}

To prove that (\ref{app:upper}) turns into an equality, 
let us consider the following observables
\begin{equation}
A_0^{(1)}=\frac{1}{\sqrt{2}}(\sigma_X+\sigma_Z),\qquad 
A_1^{(1)}=\frac{1}{\sqrt{2}}(\sigma_X-\sigma_Z)
\end{equation}
for the first observer and $A_0^{(i)}=\sigma_X$ 
and $A_1^{(i)}=\sigma_Z$ for $i=2,\ldots,N$. By a direct check one sees 
that for these observables and the graph state 
$\ket{\psi_G}$ the value of $I_G$ is exactly 
$\beta_G^Q$, which completes the proof.
\end{proof}

\section{Appendix B: Self-testing graph states}\label{app:selftesting}

In this section we provide the proof of 
Fact \ref{Fact3}, which for completeness we state formally here. 
\begin{fakt}\label{app:Fact3}
Given a graph $G$, if the corresponding Bell inequality (\ref{eq:inequalities})
is violated maximally by a state $|\psi\rangle$ and observables $\bar{A}_j^{(i)}$, then the following holds true:
\begin{equation}
\Phi[(\bar{A}^{(i_1)}_{k_{i_1}}\otimes\ldots\otimes \bar{A}^{(i_1)}_{k_{i_1}})\ket{\psi}]=
(A^{(i_1)}_{k_{i_1}}\otimes\ldots\otimes A^{(i_1)}_{k_{i_1}})\ket{\psi_G}\otimes |\mathrm{aux}\rangle, 
\end{equation}
where $\Phi=\Phi_1\otimes \ldots\otimes\Phi_N$ with $\Phi_i$ being the local isometry defined in Fig. \ref{fig:swapp}, $\ket{\mathrm{aux}}$ is some state encoding uncorrelated degrees of freedom, 
\begin{equation}
A_j^{(1)}=\frac{1}{\sqrt{2}}[\sigma_X+(-1)^{j}\sigma_Z]
\end{equation}
and
\begin{equation}
A_0^{(i)}=\sigma_X, \quad A_1^{(i)}=\sigma_Z \quad (i=2,\ldots,N).
\end{equation}
\end{fakt}

Before proving this fact, we need some
preparation. Let us consider a graph $G$ and the corresponding graph state $|\psi_G\rangle$. Let us also assume that the Bell inequality (\ref{eq:inequalities}) associated to this graph is maximally violated by a state $\ket{\psi}$ and 
observables $\bar{A}_j^{(i)}$. We then consider the following operators
\begin{equation}\label{X1Z1}
X'_1 = \frac{1}{\sqrt{2}}\left(\bar{A}_0^{(1)}+\bar{A}_1^{(1)}\right), \quad  Z'_1 = \frac{1}{\sqrt{2}}\left(\bar{A}_0^{(1)}-\bar{A}_1^{(1)}\right),
\end{equation}
and $X_1=X_1'/|X_1'|$ and $Z_1=Z_1'/|Z_1'|$. We also denote 
$X_i=A_0^{(i)}$ and $Z_i=A_1^{(i)}$ for $i=2,\ldots,N$. 
It is not difficult
to check that all these operators $X_i$ and $Z_i$
with $i=1,\ldots,N$ are unitary; for $i=2,\ldots,N$
this follows from the fact that $A_j^{(i)}$ are Hermitian 
and have eigenvalues $\pm1$, whereas for $i=1$ it stems from the polar decomposition 
(see, e.g., Ref. \cite{supic2016self}).
Let us finally choose as isometry the so-called SWAP isometry, whose output reads as follows
\begin{equation}\label{iso}
\Phi\left(\ket{+}^{\otimes N}\otimes \ket{\psi}\right)     
=\sum_{\tau \in \{0,1\}^N}\ket{\tau}\otimes \left( \bigotimes_{j=1}^NX_j^{\tau_j}Z_j^{(\tau_j)} \right) \ket{\psi},
\end{equation}
where $X_i$ and $Z_i$ are those defined above and we have also defined $Z_i^{(\tau_j)} = [\mathbbm{1}+(-1)^{\tau_j}Z_j]/2$, while the summation is over all $N$-element sequences $(\tau_1,\ldots,\tau_N)$ with each $\tau_i\in\{0,1\}$.
Notice that the action of this isometry is to perform a unitary operation $\Phi=\Phi_1\otimes\ldots\otimes\Phi_N$
on the state $\ket{+}^{\otimes N}\otimes \ket{\psi}$, where each unitary $\Phi_i$ acts on the $i$-th particle of $\ket{\psi}$ and one of the qubits in the state $\ket{+}$. A visual representation of a local branch of the isometry $\Phi_i$ is shown in Fig. \ref{fig:swapp}. We are now ready to prove Fact \ref{Fact3}. 
\begin{figure}[h!]
\centering
\includegraphics[width=0.4\textwidth]{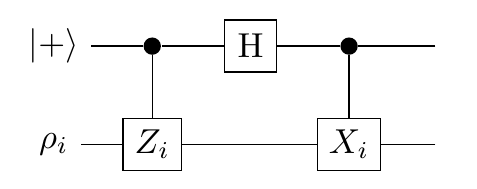}
\caption{A local branch $\Phi_i$ of the SWAP gate $\Phi=\Phi_1\otimes\ldots\otimes\Phi_N$. The isometry can be viewed as a unitary with each branch acting on  the $i$-th particle of $\ket{\psi}$ and one ancillary qubit in the state $\ket{+}$.} \label{fig:swapp}
\end{figure}
\begin{proof}
For the sake of simplicity let us assume that $2 \in n(1)$ (notice that this can always be done by relabeling the vertices). Let us also notice that, as shown 
in Ref. \cite{bamps2015sum}, one has $X_1\ket{\psi}=X_1'\ket{\psi}$ and $Z_1\ket{\psi}=Z_1'\ket{\psi}$, and so in what follows we will denote the operators in Eq. \eqref{X1Z1} by $X_1$ and $Z_1$, respectively.  

The first step of our proof is to show that 
$X_1$ and $Z_1$ as well as $X_i$ and $Z_i$ with $i=2,\ldots,N$ anticommute when acting on $\ket{\psi}$, that is, 
%
%
%
\begin{equation}\label{anticomi}
(X_iZ_i+Z_iX_i)\ket{\psi}=0\qquad (i=1,\ldots,N).
\end{equation}
To prove that (\ref{anticomi}) holds true for $i=1$ it suffices to 
use the definitions (\ref{X1Z1}). Then, to prove (\ref{anticomi}) for the rest of vertices, let us first consider 
the case $i\in n(1)$. For these vertices, the  
sum of squares decomposition (\ref{app:SOS}), 
implies the following relations 
\begin{eqnarray}
X_1\ket{\psi}&=&\bigotimes_{i\in n(1)}Z_i\ket{\psi},\nonumber\\
Z_1\ket{\psi}&=&X_{2}\otimes \left( \bigotimes_{i\in n(m)\setminus\{1\}}Z_{i} \right)  \ket{\psi},
\end{eqnarray}
%
%
which can equivalently be stated as 
\begin{eqnarray}\label{10}
X_1\otimes \left( \bigotimes_{i \in n(1)\setminus\{m\}}Z_{i} \right) \ket{\psi} &=&  Z_m \ket{\psi},\\ \label{11}
Z_1\otimes \left( \bigotimes_{i \in n(m)\setminus\{1\}}Z_{i} \right) \ket{\psi} &=& X_{m}\ket{\psi},
\end{eqnarray}
where $m\in N(1)$.
By plugging Eqs. (\ref{10}) and (\ref{11}) into Eq. (\ref{anticomi}) we have
\begin{widetext}
\begin{eqnarray}\label{ac2}
(X_mZ_m+Z_mX_m)\ket{\psi}=\left[(Z_1X_1+X_1Z_1)\otimes\bigotimes_{i\in n(1,m)}Z_i\right]\ket{\psi}=0
%
%
\end{eqnarray}
\end{widetext}
where $n(1,m)$ stands for the neighbours of 
the first and the $m$th vertex (excluding these two vertices). 
Due to the fact that, as proven before, $X_1$ and $Z_1$
anticommute, the right-hand side of the above relation
vanishes which gives us (\ref{anticomi}) for all $i\in n(1)$.

Let us then prove the anticommutation relation (\ref{anticomi}) for all vertices that are not in $n(1)$ but are neighbours of those belonging to $n(1)$.
Consider a vertex $j\notin n(1)$, which is a neighbour of a vertex $k\in n(1)$. For it the decomposition (\ref{app:SOS}) implies the following relations
\begin{equation}\label{stab1}
X_j\ket{\psi}=\left[Z_k\otimes\bigotimes_{i\in n(j)\setminus\{k\}}Z_i\right]\ket{\psi}
\end{equation}
and
\begin{equation}
Z_j\ket{\psi}=\left[X_k\otimes\bigotimes_{i\in n(k)\setminus\{j\}}Z_i\right]\ket{\psi},
\end{equation}
from which one obtains
\begin{eqnarray}
(X_jZ_j+Z_jX_j)\ket{\psi}&=&\left[(Z_kX_k+X_kZ_k)\otimes\bigotimes_{i\in n(j,k)}Z_i\right]\ket{\psi}\nonumber\\
&=&0,
\end{eqnarray}
where the last equality stems from the anticommutation
relation for $X_k$ and $Z_k$. 

Noting that there are no isolated vertices in the graph, 
we can repeat the above procedure until (\ref{anticomi}) is proven for all vertices.

Having the anticommutation relations (\ref{anticomi}) 
for all vertices of the graph, the remainder of the proof
is exactly the same as that of Theorem 4 in Ref. \cite{supic2017simple}
(see Appendix F therein). However, for completeness we present it here.

Let us go back to the action of the unitary operation $\Phi=\Phi_1\otimes\ldots\otimes\Phi_N$
on the state $\ket{+}^{\otimes N}\otimes \ket{\psi}$. Let us consider a particular term from output state given in \eqref{iso}, corresponding to the sequence $\tau$ which has $k > 0$ ones at the positions $j_1,\dots, j_k$:
\begin{equation}\label{tau}
\ket{\tau}\otimes\left(\bigotimes_{j \notin J(\tau)} Z_j^{(0)} \right)  \otimes \left( \bigotimes_{j \in J(\tau)} X_jZ_j^{(1)}\right) \ket{\psi},
\end{equation}
where $J(\tau)=\{j_1,\ldots,j_k\}$. Also, for $\tau$ let us denote by $n(\tau)$ the number of edges connecting vertices denoted by labels $j \in J(\tau)$ (without counting the same edge twice). Consider then one of the vertices from $J(\tau)$, denoted by $j_1$, and let the number of its neighbors in $J(\tau)$ be $\bar{n}(j_1)$, i.e., $\bar{n}(j_1)=|J(\tau)\cap n(j_1)|$. Due to the anticommutation relation $\{X_{j_1},Z_{j_1}\}\ket{\psi} = 0$, the expression \eqref{tau} can be rewritten as
\begin{widetext}
\begin{eqnarray}
&&\hspace{-1cm}\ket{\tau}\otimes\left( \bigotimes_{j \notin J(\tau)} Z_j^{(0)} \right) \otimes \left( \bigotimes_{j \in J(\tau)\setminus\{j_1\}} X_jZ_j^{(1)} \right) \otimes X_{j_1}Z_{j_1}^{(1)} \ket{\psi}\nonumber\\
&&=(-1)^{\bar{n}(j_1)}\ket{\tau}\otimes\left(\bigotimes_{j \notin J(\tau)} Z_j^{(0)} \right) \otimes \left( \bigotimes_{j \in J(\tau)\setminus\{j_1\}} X_jZ_j^{(1)} \right) \otimes Z_{j_1}^{(0)} \ket{\psi},
%
\end{eqnarray}
\end{widetext}
where we have also used the following relation
\begin{equation}
X_i\ket{\psi}=\bigotimes_{j\in N(i)}Z_j\ket{\psi},
\end{equation}
that stems from the sum of squares decomposition (\ref{app:SOS}) and the fact that $Z_j^{(1)} Z_j = - Z_j^{(1)}$. By using the anticommutation relations (\ref{anticomi}) as well as 
the relations (\ref{stab1}), in a similar way we can get rid of all the operators $X_j$
appearing in (\ref{tau}). This allows us to rewrite (\ref{tau}) as
\begin{equation}\label{tau1}
(-1)^{n(\tau)}\ket{\tau}\otimes \left( \bigotimes_{j=1}^N Z_j^{(0)} \right)  \ket{\psi}.
\end{equation}
After plugging the above into Eq. \eqref{iso}, one obtains
\begin{align}\label{iso1}\begin{split}
& \Phi\left(\ket{+}^{\otimes N}\otimes \ket{\psi}\right) =  \\   & =\sum_{\tau \in \{0,1\}^N}(-1)^{n(\tau)}\ket{\tau}\otimes\left( \bigotimes_{j=1}^N Z_j^{(0)} \right) \ket{\psi} \\ 
&= \ket{\psi_G}\otimes\ket{\textrm{aux}},
\end{split}
\end{align}
where we used the expression for a graph state in the computational basis
\begin{equation*}
\ket{\psi_G} = \frac{1}{\sqrt{2}^N}\sum_{\tau \in \{0,1\}^N}(-1)^{n(\tau)}\ket{\tau}.
\end{equation*}
This completes the proof. The proof for self-testing of measurements goes along the same lines as the one for the state (see for example Appendix E of \cite{supic2017simple}).

\end{proof}

\subsection{Appendix C: Fidelity bounds}

Here we adopt the techniques introduced in \cite{kaniewski2016analytic} to inequality \eqref{eq:GHZineq} to derive fidelity bounds that depend on the quantum violation observed. 

Let us recall the main ingredients of the method
from \cite{kaniewski2016analytic}, adapting them 
to our purposes.  
To this end, we consider a Bell inequality (\ref{eq:inequalities}) corresponding to a graph $G$ and a multipartite state $\rho_N$ of unspecified local dimension reaching the violation $\beta$
of it. Define the extractability
of $\ket{\psi_G}$ from $\rho_N$ as
\begin{equation}\label{extract}
\Theta (\rho_N \rightarrow \psi_G ) = \max_{\Lambda_{1} , \ldots, \Lambda_{N}} \langle \psi_G|(\Lambda_1 \otimes \ldots \otimes\Lambda_N) (\rho_N)\ket{\psi_G}, 
\end{equation}
where $F(\rho,\sigma) = \|\sqrt{\rho}\sqrt{\sigma} \|^2_1 $ is the state fidelity, $\ket{\psi_G}$ is the graph state corresponding to $G$, and the maximisation is taken over all local quantum channels $\Lambda_i$. Notice that the formulation in terms of quantum channels is equivalent to 
first adding local ancillas in any states and then performing local unitaries which extract the desired 
state $\ket{\psi_G}$ into these registers. Thus, the aim is to find the isometry under which the state $\rho_N$ is closest to the desired state $\ket{\psi_G}$. 

The aim is to put a lower bound on $\Theta (\rho_N \rightarrow \psi_G )$ in terms of the violation $\beta$. For this purpose, we notice that the fidelity in Eq. (\ref{extract})
can equivalently be written as 
\begin{equation}
\Tr[\rho_N(\Lambda_1^{\dagger}\otimes\ldots\otimes\Lambda_N^{\dagger})(|\psi_G\rangle\!\langle\psi_G|)],
\end{equation} 
where $\Lambda_i^{\dagger}$ are dual maps of the quantum channels $\Lambda_i$. Now, proving for some particular channels $\Lambda_i$  
an operator inequality 
\begin{equation}\label{OpIneq}
K:=(\Lambda_1^{\dagger}\otimes\ldots\otimes\Lambda_N^{\dagger})(|\psi_G\rangle\!\langle\psi_G|)\geq s\mathcal{B}_G+\mu\mathbbm{1}
\end{equation}
with for some $s,\mu\in\mathbbm{R}$, where $\mathcal{B}_G$
stands for the Bell operator corresponding to 
the Bell inequality (\ref{eq:inequalities}) and constructed 
from any possible dichotomic observables would imply the following inequality for the extractability
\begin{equation}
\Theta(\rho_N\to\psi_G)\geq s \beta+\mu.
\end{equation}

Proving an operator inequality (\ref{OpIneq})
for arbitrary local observables in $\mathcal{B}_G$ 
is certainly a formidable task. However, due to the fact that here we consider the simplest Bell scenario involving
two dichotomic measurements per site, one can exploit 
Jordan's lemma, which, as explained in Ref. \cite{kaniewski2016analytic} 
reduces the problem to basically  an $N$-qubit space. That is, 
the local observables $A_{x_i}^{(i)}$ can now be parametrized as
\begin{equation}\label{eq:meas}
A_{x_1}^{(1)} = \cos{\alpha_1}\, \sigma_X + 
(-1)^{x_1} \sin{\alpha_1}\, \sigma_Z, 
\end{equation}
and
\begin{equation}\label{eq:measi}
A_{x_i}^{(i)} = \cos{\alpha_i}\, \sigma_H + 
(-1)^{x_i} \sin{\alpha_i}\, \sigma_V  
\end{equation}
for $i = 2, \ldots, N$,
where $\sigma_H = (\sigma_X + \sigma_Z)/\sqrt{2}$, $\sigma_V = (\sigma_X - \sigma_Z)/\sqrt{2}$ and $\alpha_i \in [ 0,\pi/2] $. This gives rise to 
a Bell operator $\mathcal{B}_{G}(\vec{\alpha})$ that now depends on the angles $\alpha_i$. Let us then consider
particular quantum channels
\begin{equation}
\Lambda_i(x) [\rho] = \frac{1 + g(x)}{2} \rho + 
\frac{1 - g(x)}{2} \Gamma_i (x) \rho \Gamma_i (x)
\end{equation}
and the dependence on the measurement angle is encoded in the function $g(x) = (1+ \sqrt{2} )(\sin{x}+\cos{x} -1)$ together with the definition of $\Gamma(x) = M_i^{a}$ if $x \leq \pi /4$ and $\Gamma(x) = M_i^{b}$ if $x > \pi /4$. Lastly, we define $M_1^{a,b} = \sigma_X,
\sigma_Z$ and  $M_i^{a,b} = \sigma_H,\sigma_V$ for $i = 2, \ldots N$.

We now want to prove that for all possible
choices of $\alpha_i$, the following inequality is
satisfied
\begin{equation}\label{eq:fidbound}
K(\alpha_1,\ldots,\alpha_N)\geq s\mathcal{B}_G(\alpha_1,\ldots,\alpha_N)+\mu\mathbbm{1}
\end{equation}
for some choice of $s,\mu\in\mathbbm{R}$, which, as exampled earlier, would imply 
an inequality for the extractability. 

We have performed numerical tests to derive bounds of the kind \eqref{eq:fidbound} for the inequality for the GHZ state and the ring cluster state for values of $N \leq 7$.
The applied procedure works as follows: given a fixed $s$, estimate the corresponding $\mu$ by numerically computing the minimal eigenvalue of the operator $K + s \mathcal{B}_G$ and minimizing over all the angles $\alpha_i$.
Notice that to have a fidelity bound that leads to fidelity $1$ at the point of maximal violation, the inequality \eqref{eq:fidbound} has to become tight for the measurements angles leading to the maximal violation, that is $\alpha_i = \pi/4$ for all $i = 1, \ldots, N$.
As a second step, we therefore estimated numerically the minimum value of $s$ for which the corresponding bound still satisfied such a property. This led to linear bounds with the optimal slope.

\section{Appendix D: Increasing the quantum violation}\label{app:rotation}

Here we explain in more detail how our Bell inequalities can be modified
to allow for higher ratios $\beta_G^Q/\beta_G^C$. 

Given a Bell inequality (\ref{eq:inequalities}) corresponding to a graph $G$, consider
a vertex $j\in V$ that neither belongs to $n(1)$ nor it shares a neighbour 
with the first vertex. Then one can apply a second substitution ${\sigma_X}_j\to A_0^{(j)} + A^{(j)}_1$ and ${\sigma_Z}_j \rightarrow  A_0^{(j)} - A^{(j)}_1$ at that vertex. This gives us the following Bell inequality
\begin{widetext}
\begin{eqnarray}\label{eq:inequalities_2}
I_G:&=&n_{\max}\left\langle (A_0^{(1)}+A_1^{(1)})\prod_{i\in n(1)}A_{1}^{(i)}\right\rangle+\sum_{i\in n(1)}
\left\langle (A_0^{(1)}-A_1^{(1)})A_0^{(i)}\prod_{j\in n(i)\setminus\{1\}}A_1^{(j)}\right\rangle\nonumber\\
&&+n_{j}\left\langle (A_0^{(j)}+A_1^{(j)})\prod_{i\in n(j)}A_{1}^{(j)}\right\rangle+
\sum_{i\in n(j)}
\left\langle (A_0^{(j)}-A_1^{(j)})A_0^{(i)}\prod_{k\in n(i)\setminus\{j\}}A_1^{(k)}\right\rangle\nonumber\\
&&+\sum_{i\notin n(1)\cup\{1\}\cup n(j)\cup\{j\}}
\left\langle A^{(i)}_0 \prod_{k\in n(i)} A^{(j)}_{1}\right\rangle\leq \beta_G^C,
\end{eqnarray}
\end{widetext}
for which, as before, it is not difficult to analytically compute its maximal quantum and classical values. They read 
\begin{equation}
\beta^{(2)}_{G,Q} = N + n_{max} + n(j) -2
\end{equation} 
and
\begin{equation}
\beta^{(2)}_{G,C} = (2\sqrt{2} -1) [n_{max} + n(j)] + N -2,
\end{equation}
respectively, where the superscript indicates the fact we have played our trick with two sites. It then follows that $\beta^{(2)}_{G,Q}/\beta^{(2)}_{G,C} \geq \beta^{Q}_{G}/\beta^{C}_{G}$ for any graph $G$. 

We can repeat the same procedure with any other vertex which does not belong to $n(1)$ nor $n(j)$ and does not share any neighbour with neither vertex $1$ nor $j$. In such a way we can increase the ratio again. 
This clearly comes at the cost of increasing the number of expectation values appearing
in the Bell expression. Notice, however, that their linear scaling with $N$ is preserved even in the case in which the above mentioned substitution is applied to any available vertex.
Indeed, consider the extremal case in which the replacement with $A_0^{(i)} \pm A_1^{(i)}$ appears for some party $i$ in all the terms in the sum in \eqref{eq:inequalities}. Since such terms correspond to the $N$ generators $G_i$ of the stabilizer group, the resulting amount of correlators is exactly $2N$.
One can use a similar argument to see that the ratio $\beta^{Q}_{G}/\beta^{C}_{G}$ is always bounded and cannot exceed $\sqrt{2}$.

\begin{figure}[h!]
\centering
\includegraphics[scale = 0.4]{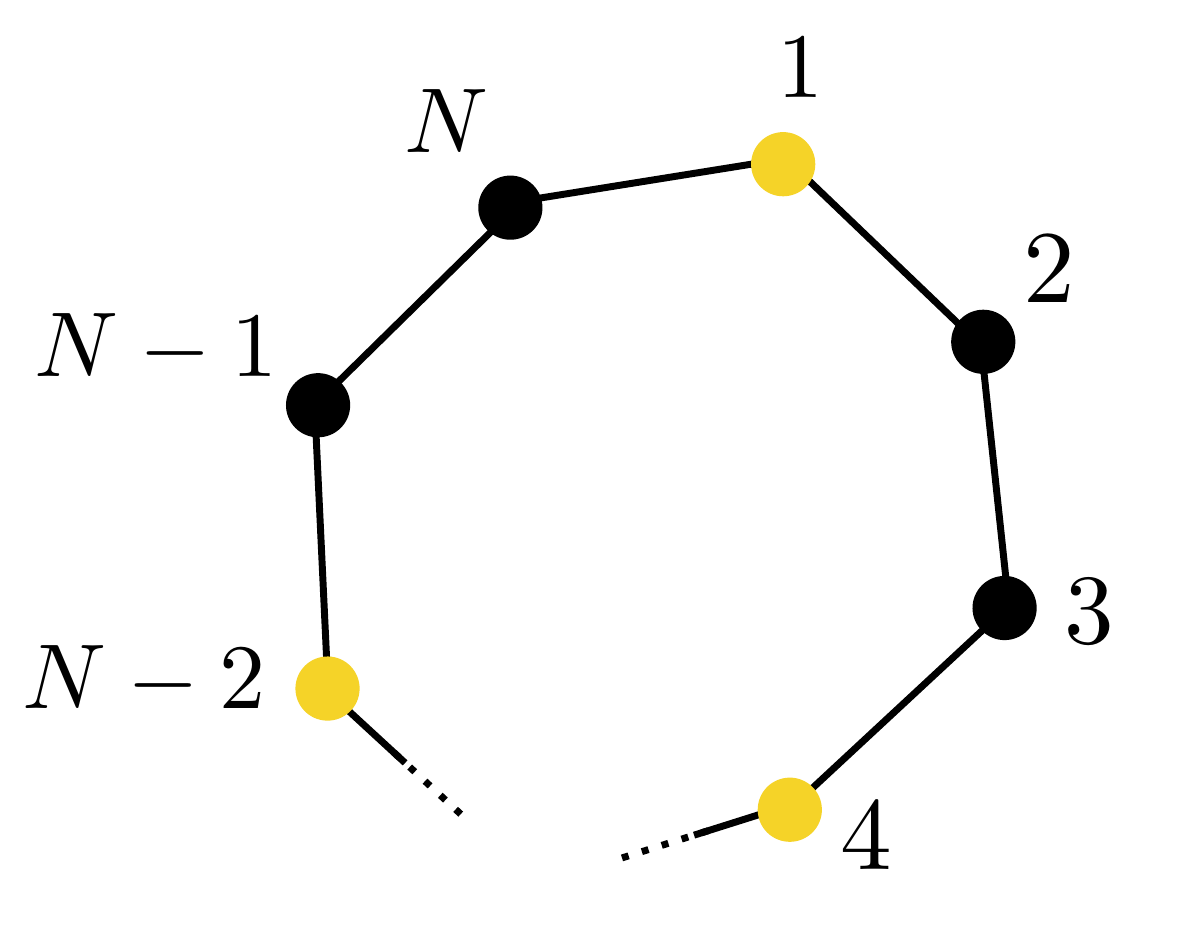}
\caption{Pictorial representation of the method to generate Bell inequalities with higher quantum violation, taking the ring graph state. The vertices coloured in yellow are the ones for which the substitution ${\sigma_X}_j\to A_0^{(j)} + A^{(j)}_1$ and ${\sigma_Z}_j \rightarrow  A_0^{(j)} - A^{(j)}_1$ is applied.} 
\label{fig:exampleapp}
\end{figure}

As an illustrative example let us turn back to the inequality \eqref{eq:ring} for the ring graph state. To improve the quantum-classical ratio, our methods tells us to choose \added{a} vertex whose neighbourhood does not interesect with $n(1) \cup \lbrace 1 \rbrace $. As a candidate to apply the second substitution we take vertex $4$ so to obtain the following Bell inequality
\begin{widetext}
\begin{eqnarray}
I^{(2)}_{\mathrm{ring}}&:=&2\langle A_1^{(N)}(A_0^{(1)}+A_1^{(1)})A_1^{(2)}\rangle+
\langle (A_0^{(1)}-A_1^{(1)})A_0^{(2)}A_1^{(3)}\rangle
+\langle A_1^{(N-1)}A_0^{(N)}(A_0^{(1)}-A_1^{(1)})\rangle \nonumber \\  
&& + 2\langle A_1^{(3)}(A_0^{(4)}+A_1^{(4)})A_1^{(5)}\rangle + \langle (A_0^{(4)}-A_1^{(4)})A_0^{(5)}A_1^{(6)}\rangle
+\langle A_1^{(2)}A_0^{(3)}(A_0^{(4)}-A_1^{(4)})\rangle\nonumber\\
&&+\sum_{i=6}^{N-1}\langle A_1^{(i-1)}A_0^{(i)}A_1^{(i+1)} \rangle\leq N+2,
\end{eqnarray}
\end{widetext}
with a corresponding quantum violation of $\beta^{(2)}_Q = N - 6 + 8\sqrt{2}$.
Let us then notice that we can repeat our trick for all vertices $3i+1$ for $i=1,\ldots, \lfloor N/3 \rfloor$ as each pair  of them does not belong to each others' neighborhoods nor 
shares a common vertex (see \figref{fig:exampleapp} for a pictorial representation). Thus we can generate a series of $ \lfloor N/3 \rfloor$ Bell inequalities whose classical and quantum values can easily be computed and are given by
\begin{equation}
\beta^C_k=N+k,\qquad \beta^Q_k=N+(4\sqrt{2}-3)k
\end{equation}
with $k=1,\ldots, \lfloor N/3 \rfloor$. The 
sequence satisfies $\beta^Q_{k+1}/\beta^C_{k+1}>\beta^Q_k/\beta^C_k$ for any $k$ and 
the ratio attains its maximal value for $N = 3L$, $k=L$, amounting to exactly $\beta^Q_N/\beta^C_N=\sqrt{2}$. The resulting inequality read as follows
\begin{eqnarray}
I^{\mathrm{max}}_{\mathrm{ring}}&:=& \sum_{i = 1}^{L} 2\langle A_1^{(3i)}(A_0^{(3i +1)}+A_1^{(3i +1)})A_1^{(3i +2)}\rangle \nonumber \\ 
&& + \langle (A_0^{(3i+1)}-A_1^{(3i+1)})A_0^{(3i +2)}A_1^{(3i + 3 )}\rangle\nonumber\\
&& +\langle A_1^{(3i-1)}A_0^{(3i)}(A_0^{(3i+1)}-A_1^{(3i+1)}) \rangle .
\end{eqnarray}

\section{Appendix E: Self-testing the partially entangled GHZ state from its stabilizers}

Here we look at the method used to derive Bell inequalities for graph states as a more general strategy, so to apply it to other states as well. We will show how to do it for the partially entangled GHZ state of the following form
\begin{equation}\label{partiallyEntGHZ}
|\mathrm{GHZ}_N(\theta)\rangle=\cos\theta\ket{0}^{\otimes N}+\sin\theta\ket{1}^{\otimes N}.
\end{equation}
For this state it is possible to define $N$ independent stabilizing operators:
\begin{equation}\label{Stab1App}
S_1  = \sin{2\theta} {\sigma_X}_1 {\sigma_X}_2 \ldots {\sigma_X}_N +  \cos{2 \theta} {\sigma_Z}_1 
\end{equation}
for the first site, and
\begin{equation}
S_i =  {\sigma_Z}_1 {\sigma_Z}_i 
\end{equation}
for sites $i = 2,\ldots,N$. Indeed, one can verify that 
$S_i |\mathrm{GHZ}_N(\theta)\rangle =|\mathrm{GHZ}_N(\theta)\rangle$ for any $\theta \in [0,\pi/4]$ and $i=1,\ldots,N$.

We will start by showing how to generalize the self-testing method introduced in Appendix B 
for the graph states and then building on that we will derive Bell inequalities for 
(\ref{partiallyEntGHZ}).


\textit{Self-testing proof.} Let us begin by making the following substitutions 
\begin{equation}\label{eq:rotate}
X'_1 = \frac{A^{(1)}_0 + A^{(1)}_1}{2\sin{\mu}}, \qquad {Z'}_1 = \frac{A^{(1)}_0 - A^{(1)}_1}{2\cos{\mu}},
\end{equation}
with their regularized versions being $X_1 = X'_1/|X'_1|$, $Z_1 = Z'_1/|Z'_1|$  and $X_i = A^{(i)}_0$, $Z_i = A^{(i)}_1$ for $i = 2,\ldots,N$.
Notice that the operators for the first observer anticommute by construction, while all the remaining ones square to identity, that is, $X_i^2={Z}_i^{2}=\mathbbm{1}$. 

Suppose now that we are given a Bell expression $\mathcal{I}$ whose maximal quantum value $\beta_Q$ is achieved by a state $\ket{\psi}$. Let us assume, moreover, that the corresponding Bell operator $\mathcal{B}$ admits the following sum of squares
\begin{equation} \label{eq:stabSOS}
c (\beta_Q \mathbbm{1}  - \mathcal{B} ) = \sum_{i = 1}^N \alpha_i^2 ( \mathbbm{1} - \tilde{S}_i )^2, 
\end{equation}
where we identify with $\tilde{S}_i$ the stabilizer operators with the substituted operators 
${X}_i,{Z}_i$. Such a decomposition would imply that the state $\ket{\psi}$ satisfies the stabilizing conditions 
\begin{equation}\label{StabTilde}
\tilde{S}_i \ket{\psi } = \ket{\psi }
\end{equation}
with $i = 1,\ldots,N$.

We now proceed to show that, with any choice of operators of the kind of \eqref{eq:rotate}, the above two equations suffice to self-test the partially entangled GHZ state for any $\theta \in (0, \pi/4 ] $. 

First, let us see how the stabilizing conditions allow to prove that all the pairs ${X}_i,{Z}_i$ anticommute and square to identity when acting on the state. 

Let us begin with ${X}_1$ and ${Z}_1$. First, from the definitions (\ref{eq:rotate}) we directly see that 
\begin{equation}\label{antiTilde1}
\{{X}_1,{Z}_1\}=0.
\end{equation}
Then, from the conditions (\ref{StabTilde}) and the fact that ${Z}_i^2=\mathbbm{1}$ for any $i=2,\ldots,N$ we immediately obtain ${Z}_1\ket{\psi}={Z}_i\ket{\psi}$, which implies that 
\begin{equation}
{Z}_1^2\ket{\psi}={Z}_1{Z}_i\ket{\psi}=\tilde{S}_i\ket{\psi}=\ket{\psi},
\end{equation}
and, as a result, that $\tilde{S}_i^2\ket{\psi}=\ket{\psi}$. To finally prove 
that ${X}_i^2=\mathbbm{1}$, we rewrite
(\ref{eq:rotate}) as
\begin{equation}
{X}_1=\frac{1}{\sin 2\theta}(\tilde{S}_1-\cos 2\theta {Z}_1)\mathsf{X}_1,
\end{equation}
where $\mathsf{X}_1={X}_2\ldots {X}_N$. Due to the fact that $\mathsf{X}^2_1=\mathbbm{1}$, we then have
\begin{equation}\label{X12}
{X}_1^2=\frac{1}{\sin^2 2\theta}\left[\tilde{S}_1^2-\cos 2\theta \{\tilde{S}_1,{Z}_1\}+\cos^22\theta {Z}_1^2\right].
\end{equation}
From the very definition of $\tilde{S}_1$ we can rewrite the anticommutator
appearing in the above as
\begin{eqnarray}\label{Oshee}
\{\tilde{S}_1,{Z}_1\}&=&\sin2\theta\{{X}_1,{Z}_1\}\mathsf{X}_1+2\cos2\theta {Z}_1^2\nonumber\\
&=&2\cos2\theta {Z}_1^2,
\end{eqnarray}
where the second equality stems from the anticommutation relation (\ref{antiTilde1}).
The identity (\ref{Oshee}) allows us to simplify Eq. (\ref{X12}) as
\begin{equation}\label{X12_new}
{X}_1^2=\frac{1}{\sin^2 2\theta}\left(\tilde{S}_1^2-\cos^2 2\theta {Z}_1^2\right), 
\end{equation}
which, due to the fact that $\tilde{S}_1^2\ket{\psi}={Z}_1^2\ket{\psi}=\ket{\psi}$, directly implies that ${X}_1^2\ket{\psi}=\ket{\psi}$.

Let us now turn to the operators  ${X}_i$ and ${Z}_i$ for the remaining sites 
$i=2,\ldots,N$. We have already noticed that ${X}_i^2={Z}_i^2=\mathbbm{1}$, so in what follows we prove that they anticommute. With the aid of Eq. (\ref{Stab1App}) we
can express ${X}_i$ as
\begin{equation}
{X}_i=\frac{1}{\sin2\theta}\mathsf{X}_i\left(\tilde{S}_1-\cos2\theta {Z}_1\right),
\end{equation}
where $\mathsf{X}_i={X}_1\ldots {X}_{i-1}{X}_{i+1}\ldots {X}_N$. This, after some straightforward maneuvers, allows us to write
\begin{equation}
\{{X}_i,{Z}_i\}\ket{\psi}=\frac{1}{\sin2\theta}\mathsf{X}_i
\left[\{\tilde{S}_1,{Z}_1\}-2\cos2\theta {Z}_1^2\right]\ket{\psi}
\end{equation}
To see that the right-hand side of the above equation vanishes it suffices to 
use Eq. (\ref{Oshee}).
We have thus established that 
\begin{equation}
\{{X}_i,{Z}_i\}\ket{\psi}=0
\end{equation}
as well as ${X}_i^2\ket{\psi}={Z}_i^2\ket{\psi}=\ket{\psi}$ for all $i=1,\ldots,N$.
Let us now use them to prove our self-testing statement with the isometry $\Phi=\Phi_1\otimes\ldots\otimes \Phi_N$ with each $\Phi_i$ traditionally defined as
in Fig. \ref{fig:examples}. As above, each operator $\Phi_i$ acts on one of the particles of state $\ket{\psi}$ and a qubit state $\ket +$, giving
\begin{equation}\label{iso_new}
\Phi\left(\ket{+}^{\otimes N}\otimes \ket{\psi}\right)     
=\sum_{\tau \in \{0,1\}^N}\ket{\tau}\otimes \left( \bigotimes_{j=1}^N{X}_j^{\tau_j}{Z}_j^{(\tau_j)} \right)\ket{\psi},
\end{equation}

Let us first show that all terms in (\ref{iso_new}) except for 
$\tau=(0,\ldots,0)$ and $\tau=(1,\ldots,1)$ vanish. To this end, 
consider a sequence $\tau$ in which $\tau_{m}=0$ and $\tau_{n}=1$
for some $m\neq n$. For such a sequence we can 
rewrite the corresponding term in (\ref{iso_new}) as
%
\begin{eqnarray}\label{dupa}
\hspace{-3cm}&&\left ( \bigotimes_{j \neq m,n} {X}_j^{\tau_j}{Z}_j^{(\tau_j)} \right) \otimes {Z}_{m}^{(0)} {X}_{n} {Z}_{n}^{(1)} \ket{\psi} \nonumber \\ 
&&= \left ( \bigotimes_{j \neq j_1,j_2} {X}_j^{\tau_j}{Z}_j^{(\tau_j)} \right) \otimes  {X}_{n} {Z}_{1}^{(1)} {Z}_{1}^{(0)} \ket{\psi},
\end{eqnarray}
where we used the anticommutation relation for ${X}_n$ and ${Z}_n$
as well as the fact that ${Z}_i\ket{\psi}={Z}_1\ket{\psi}$ for $i=2,\ldots,N$.
Noticing then that ${Z}_{1}^{(1)} {Z}_{1}^{(0)}=0$ as both ${Z}_{1}^{(i)}$
are unnormalized projections onto orthogonal subspaces, we see that 
(\ref{dupa}) amounts to zero.

Hence, the expression \eqref{iso_new} reduces to the following two terms
\begin{widetext}
\begin{eqnarray}
\Phi\left(\ket{+}^{\otimes N}\otimes \ket{\psi}\right) & =& \ket{0}^{\otimes N}\otimes \left({Z}_1^{(0)} \ldots  {Z}_N^{(0)}\right) \ket{\psi} +
\ket{1}^{\otimes N}\otimes \left({X}_1 {Z}_1^- \ldots  {X}_N {Z}_N^-\right) \ket{\psi} \nonumber \\
& =& \ket{0}^{\otimes N}\otimes ({Z}_1^{(0)})^N  \ket{\psi} + \ket{1}^{\otimes N}\otimes  \left[{X}_1 ({Z}_1^{(1)})^{N} {X}_2 \ldots {X}_N \right] \ket{\psi} \nonumber \\
& =& \ket{0}^{\otimes N}\otimes {Z}_1^{(0)}  \ket{\psi} + \ket{1}^{\otimes N}\otimes  \left[{Z}_1^{(0)}{X}_1 \ldots {X}_N \right] \ket{\psi},
\end{eqnarray}
where to obtain the second equality we exploited conditions (\ref{StabTilde})
for all $i=2,\ldots,N$, whereas the second equality follows from the fact that
$[{Z}_1^{(j)}]^2\ket{\psi}={Z}_1^{(j)}\ket{\psi}$ for $j=0,1$ and the
anticommutation relation (\ref{antiTilde1}). Using then (\ref{Stab1App}), the above
can be rewritten as 
\begin{eqnarray}
\Phi\left(\ket{+}^{\otimes N}\otimes \ket{\psi}\right) &=& \frac{1}{\sin 2\theta} \left[ \sin 2\theta  \ket{0}^{\otimes N}\otimes {Z}_1^{(0)} \ket{\psi} + \ket{1}^{\otimes N}\otimes ( {Z}_1^{(0)}\tilde{S}_1   - \cos 2\theta {Z}_1^{(0)} {Z}_1 )\ket{\psi} \right]  \nonumber \\
&=&\frac{1}{\sin 2\theta} \left[  \sin2 \theta \ket{0}^{\otimes N}  + (1 - \cos2 \theta) \ket{1}^{\otimes N} \right]\otimes {Z}_1^{(0)}  \ket{\psi} \nonumber\\
 &=& \ket{\mathrm{aux}}\otimes \ket{\mathrm{GHZ}_N (\theta )},
\end{eqnarray}
\end{widetext}
where $\ket{\mathrm{aux}} = (1/\cos \theta) {Z}_1^{(0)} \ket{\psi}$. To obtain the second equality we used the facts that $\tilde{S}_1$ stabilizes $\ket{\psi}$ and that ${Z}_1^{(0)}{Z}_1={Z}_1^{(0)}$, while the last one is a consequence of the two well-known trigonometric relations $\sin2\theta=2\sin\theta\cos\theta$ and 
$1-\cos2\theta=2\sin^2\theta$. This completes our self-testing statement.

\noindent \textit{Deriving the Bell inequality.} Now, what is left to show is that \eqref{eq:stabSOS} can indeed be satisfied and thus give rise to a non-trivial Bell inequality. We will see that this can be done by choosing the free angle $\mu$ and the $\alpha_i$ parameters accordingly. To do so, let us first compute the square of the stabilizing operators 
\begin{eqnarray}
\tilde{S}_1^2  &=&\frac{1}{2}\left(\frac{\sin^2 2\theta}{\sin^2\mu}+\frac{\cos^22\theta}{\cos^2\mu}\right)\mathbbm{1}\nonumber\\
&&+\frac{1}{4}\left(\frac{\sin^2 2\theta}{\sin^2\mu}-\frac{\cos^22\theta}{\cos^2\mu}\right)\{A_0^{(1)},A_1^{(1)}\}
\end{eqnarray}
and 
\begin{equation}
\tilde{S}_i^2 = \frac{1 }{2 \cos^2{\mu}}\mathbbm{1} - \frac{1 }{4 \cos^2{\mu}} \lbrace A^{(1)}_0 , A^{(1)}_1 \rbrace
\end{equation}
for $i = 2,\ldots,N$. With these identities the sum of squares 
\eqref{eq:stabSOS} can be expanded as
\begin{widetext}
\begin{eqnarray}
\sum_{i = 1}^N \alpha_i^2 ( \mathbbm{1} - \tilde{S}_i )^2 
& =& \sum_{i = 1}^N \alpha_i^2 \mathbbm{1} - 2 \sum_{i = 1}^N \alpha_i^2  \tilde{S}_i   + 
\frac{1}{2}\left[  \left( \frac{\sin^2{2\theta}}{ \sin^2{\mu}} + \frac{\cos^2{2\theta}}{ \cos^2{\mu}} \right)\alpha_1^2 + 
\frac{1}{ \cos^2{\mu}}\sum_{i=2}^N{\alpha_i^2}\right] \mathbbm{1} \\
&& + \frac{1}{4}  \left[ \left(  \frac{\sin^2{2\theta}}{ \sin^2{\mu}} - \frac{\cos^2{2\theta}}{ \cos^2{\mu}} \right)\alpha_1^2 
- \frac{1}{\cos^2{\mu}}  \sum_{i=2}^N{\alpha_i^2}\right]\lbrace A_0^{(1)} , A_1^{(1)} \rbrace.
\end{eqnarray}
\end{widetext}
Now, we want the term standing in front of the anticommutator to vanish. This can be done by setting: $\alpha_1^2 =\sqrt{2}(N - 1)$ and $\alpha_2^2 = \ldots = \alpha_N^2 =  \sqrt{2} $ and the angle $\mu$ so that $2\sin^2{\mu} = \sin^2{2\theta}$. This gives
\begin{widetext}
\begin{equation}
\sum_{i = 1}^N \alpha_i^2 ( \mathbbm{1} - \tilde{S}_i )^2 = 
2 \left\{2\sqrt{2} (N-1) \mathbbm{1} -  \left[ (N-1)\sqrt{2} \tilde{S}_1  + \sqrt{2}\sum_{i =2}^N  \tilde{S}_i\right] \right\},
\end{equation}
\end{widetext}
%
%
where we keep the $\sqrt{2}$ factor inside the curly brackets 
for further convenience. 
We can thus identify $\beta_Q = 2\sqrt{2} (N-1)$ 
and the remaining terms appearing on the left-hand side of the above 
as the Bell operator
\begin{equation}
\mathcal{B}=(N-1)\sqrt{2}\tilde{S}_1+\sqrt{2}\sum_{i=2}^N\tilde{S}_i.
\end{equation}
This, after substituting the expressions of the operators 
${X_i},{Z_i}$ in terms of arbitrary observables for all $i$, 
leads us to the following Bell inequality
\begin{widetext}
\begin{eqnarray}\label{eq:othertilted}
\mathcal{I}_{
\theta}&:=&(N-1)\langle(A_0^{(1)}+A_1^{(1)})A_0^{(2)}\ldots A_0^{(N)}\rangle  +(N-1)\frac{\cos{2\theta}}{\sqrt{1 + \cos^2{2\theta}}}(\langle A_0^{(1)}\rangle-\langle A_1^{(1)}\rangle)\nonumber\\
&&+\frac{1}{\sqrt{1 + \cos^2{2\theta}}}\sum_{i=2}^N \langle (A_0^{(1)}- A_1^{(1)})A_1^{(i)}\rangle\leq \beta_C,
\end{eqnarray}
\end{widetext}
where $\beta_C$ is the classical bound that we compute below.
For this purpose, we can optimize $\mathcal{I}_{\theta}$ over all the deterministic strategies corresponding to the different choices $A_{x_i}^{(i)} = \pm 1$. Given the simple form of the inequality, we can divide into the two subcases $A_0^{(1)} = \pm A_1^{(1)}$ and notice that the maximum is attained in the case in which the observables of the first party take opposite signs, which results in
\begin{equation}
\beta_C(\theta) =  2(N-1) \frac{1 + \cos{2\theta}}{\sqrt{1 + \cos^2{2\theta}}}.
\end{equation}
Notice that $\beta_C(\pi/4) = 2(N-1)$ and we recover the limit case of the GHZ state and inequality \eqref{eq:GHZineq}, while for $\theta=0$ one has $\beta_C(0) = 2\sqrt{2} (N-1)$ and there is obviously no quantum violation. Moreover one can see that $\beta_C(\theta)$ is a decreasing function of $\theta$ in the considered interval. This implies that \eqref{eq:othertilted} is violated for any value of $\theta$ in the given interval.
Interestingly, in the case $N = 2$ we obtain a self-testing inequality for the partially entangled two-qubit state that is inequivalent to the known tilted CHSH~\added{\cite{acinmassarpironio,bamps2015sum}}. .

\end{document}